\documentclass{llncs}
\usepackage{amsmath, cite}
\usepackage{latexsym}
\usepackage{amssymb}
\usepackage{amsfonts}
\usepackage{graphicx}
\usepackage{url}
\usepackage{color}
\usepackage{authblk}
\usepackage{setspace}

\usepackage{epic,eepic}

\newtheorem{Theorem}{Theorem}
\newtheorem{Lemma}{Lemma}

\newtheorem{Proposition}{Proposition}

\makeatletter
\newcommand{\Extend}[5]{\ext@arrow 0099{\arrowfill@#1#2#3}{#4}{#5}}
\makeatother
\title{Computing the Least-core and Nucleolus for \\ Threshold Cardinality Matching Games}

\author{Qizhi Fang\inst{1} \and Bo Li\inst{1} \and Xiaoming Sun\inst{2} \thanks{Email: \email{sunxiaoming@ict.ac.cn}} \and Jia Zhang\inst{2} \and Jialin Zhang\inst{2}}
\institute{School of Mathematical Sciences, Ocean University of China, Qingdao, China\and Institute of Computing Technology, Chinese Academy of Sciences, Beijing, China}

\begin{document}
\maketitle
\pagestyle{plain}
\bibliographystyle{plain}
\begin{abstract}
    Cooperative games provide a framework for fair and stable profit allocation in multi-agent systems. \emph{Core}, \emph{least-core} and \emph{nucleolus} are such solution concepts that characterize stability of cooperation. In this paper, we study the algorithmic issues on the least-core and nucleolus of threshold cardinality matching games (TCMG). A TCMG is defined on a graph $G=(V,E)$ and a threshold $T$, in which the player set is $V$ and the profit of a coalition $S\subseteq V$ is 1 if the size of a maximum matching in $G[S]$ meets or exceeds $T$, and 0 otherwise. We first show that for a TCMG, the problems of computing least-core value, finding and verifying least-core payoff are all polynomial time solvable. We also provide a general characterization of the least core for a large class of TCMG. Next, based on Gallai-Edmonds Decomposition in matching theory, we give a concise formulation of the nucleolus for a typical case of TCMG which the threshold $T$ equals $1$. When the threshold $T$ is relevant to the input size, we prove that the nucleolus can be obtained in polynomial time in bipartite graphs and graphs with a perfect matching.

\end{abstract}
    \thispagestyle{plain}
    \section{Introduction}

One of the important problems in cooperative game is how to
distribute the total profit generated by a group of agents to
individual participants. The prerequisite here is to make all the
agents work together, {\it i.e.}, form a grand coalition. To achieve
this goal,  the collective profit should be distributed properly so
as to minimize the incentive of subgroups of agents to deviate and
form coalitions of their own. This intuition is formally captured by
several solution concepts, such as the {\it core}, the {\it least-core}, and the
{\it nucleolus}, which will be the focus of this paper.

The algorithmic issues in cooperative games are especially
interesting since the definitions of many solution concepts would
involve in an exponential number of constraints~\cite{Schmeidler:1969}.
Megiddo~\cite{megiddo1978computational} suggested that finding a solution
should be done by an efficient algorithm (following
Edmonds~\cite{edmonds1965}), {\it i.e.}, within time polynomial in the
number of agents. Deng and Papadimitriou~\cite{Deng:1994}
suggested the computational complexity be taken into consideration as
another measure of fairness for evaluating and comparing different
solution concepts. Subsequently, various interesting complexity
and algorithmic results have been investigated.
On one hand, efficient algorithms have
been proposed for computing the core,
the least-core and the nucleolus, such as, for assignment
games\cite{Solymosi:1994}, standard tree
games\cite{granot1996kernel}, matching games\cite{Kern:2003},
airport profit games \cite{branzei2006simple}, spanning connectivity
games \cite{aziz2009wiretapping}, flow games\cite{Deng:2009} and
weighted voting games\cite{Elkind:2009}. On the other hand, some
negative results are given. For example, the problems of computing the nucleolus
and testing whether a given distribution belongs to the core or the nucleolus
are proved to be NP-hard for minimum spanning
tree games\cite{faigle1998note},  flow games and linear production
games \cite{Deng:2009}.


Matching game is one of the most important combinatorial cooperative
games which has attracted much attention
\cite{Deng:1999, shapley1971assignment,Solymosi:1994,Kern:2003,Biro:2012,Chen:2012}.
Assignment game
is a special case of matching game defined on a bipartite graph, which is
introduced by Shapley and Shubik\cite{shapley1971assignment}
to formulate two-sided markets. In this case, the core is
always non-empty, and the nucleolus can be found in polynomial time
\cite{Solymosi:1994}. For matching games defined on  general graphs,
Deng, et al.,\cite{Deng:1999} gave a
sufficient and necessary condition on the non-emptiness of the core.
Kern and Paulusma\cite{Kern:2003} presented an efficient algorithm for computing
the nucleolus for cardinality matching games
based on a polynomial description of the least-core of such games. Recently,
Bir\'o, Kern, and Paulusma\cite{Biro:2012} generalized the work of
\cite{Solymosi:1994} to develope an efficient algorithm for computing
the nucleolus for matching games on weighted graphs when the core is nonempty.
Chen, Lu and Zhang\cite{Chen:2012} further discussed the fractional matching games.
However, for the matching games defined on general weighted graphs,
the computational complexity on the least-core and the nucleolus is still open.

We follow the stream and study the least-core and
nucleolus of a natural variation of matching games, called
threshold matching games~\cite{Aziz:2010}. In this game model,
each agent controls a vertex in an edge-weighted graph,
and a coalition wins only if the maximum weighted matching
in the induced subgraph meets or exceeds a given threshold value.
Although Haris et al~\cite{Aziz:2010} proved that
computing the least-core and the nucleolus for
threshold matching games defined on general weighted graphs is NP-hard,
the related algorithmic problems have not been discussed when
restricted to unweighted graphs.

In this paper, we aim to compute the least-core and the nucleolus
for the threshold matching games on unweighted graph,
especially when the core is empty.
Firstly, we show that
for an arbitrary threshold value, the least-core can be obtained in polynomial time
through separation oracle technique.
By linear program duality, we further provide a general characterization of the least
core for a large class of  threshold cardinality matching games,
which can be used to simplify the sequence of linear programs of the nucleolus.
Secondly, we discuss the algorithms for the nucleolus. In the case that
the threshold is independent of input size, the nucleolus can be found
in polynomial time. Especially, when the threshold being
one (which is called \emph{edge coalitional games}), we know that finding the least-core
and the nucleolus can be done efficiently based on a clear description
of the least-core.  When the threshold value
is relevant to the input size, we prove that the least-core
and the nucleolus can also be computed in polynomial time for the games
on two typical graphs, the graphs with a perfect matching or bipartite graphs.
To our surprise, in all the cases considered, the least-core and the nucleolus
do not depend on the value of the threshold.
We conjecture our method can be generalized into dealing with general graphs.
Besides, we establish the relationship between the least-core of a threshold matching
game and the mixed Nash Equilibrium of a non-cooperative two-person zero-sum game,
called {\it matching intercept game}.

The organization of the paper is as follows.
In section \ref{sec:name:preliminaries},
we introduce some concepts in cooperative game theory, and
the definitions of  \emph{threshold matching game} (TMG)
and \emph{threshold cardinality matching game} (TCMG).
In section \ref{sec:name:LC}, we discuss the least-core for TCMGs.
Section \ref{sec:name:n1}, section \ref{sec:name:n2} and section \ref{sec:name:n3} are dedicated to
the efficient algorithms for computing the nucleolus of \emph{edge coalitional games} (ECG),
TCMG \emph{on graphs with a perfect matching}, and {TCMG \emph on bipartite graphs}.


    \section{Preliminary and Definition \label{sec:name:preliminaries}}

\subsection{Cooperative Game Theory}\label{sec:cooperativegametheory}

A cooperative game $\Gamma=(N,v)$ consists of a player set
$N=\{1,2,\cdots,n\}$  and a value function $v: 2^N\rightarrow R$
with $v(\emptyset)=0$. $\forall S\subseteq N$,
$v(S)$ represents the profit obtained by $S$ without the help of others.

A game $\Gamma=(N,v)$ is {\em monotone} if $v(S')\leq v(S)$ whenever $S'\subseteq S$.
A simple game is a monotonic game with $v:2^N\rightarrow \{0,1\}$ such that
$v(\emptyset)=0$ and $v(N)=1$. A coalition $S\subseteq N$ is {\em winning}
if $v(S)=1$, and {\em losing} if $v(S)=0$.
A player $i$ is called a {\em veto} player if he belongs to
all winning coalitions. It is easy to see that $i$ is a veto player if and only if
$v(N)=1$ but $v(N\setminus \{i\})=0$.

In cooperative games, we focus on how to distribute the total profit $v(N)$
in a fair or stable way.
Different requirements on fairness and stability
lead to different kinds of distributions, which are generally referred to
{\em{solution concepts}}. In this paper, we focus on three
solutions based on stability:
the {\em{core}}, the {\em{least-core}} and the {\em{nucleolus}}.

Given a cooperative game $\Gamma=(N,v)$, we use $x=(x_1,x_2,\cdots,x_n)$
to represent the payoff vector while $x_i$ is the payoff for player $i$.
For convenience, let $x(S)\triangleq \sum_{i\in S}x_i$. We say a payoff
vector $x$ is an \emph{imputation} of $\Gamma$, if $x(N) =v(N)$
and $\forall i\in N$, $x_i \geq v(\{i\})$ (individual rationality).
The {\em core} of $\Gamma$  is defined as:
$${\cal C}(\Gamma):=\{x\in R^n:x(N)=v(N) \ \mbox{and} \ x(S) \geq v(S),
\forall S\subseteq N\}.$$
A  payoff vector in ${\cal C}(\Gamma)$ guarantees that any coalition $S$ can not
get more profit if it breaks away from the grand coalition. This is called
group rationality.

When ${\cal C}(\Gamma) = \emptyset$, there is a nature
relaxation of the core: the {\em least-core}.
Given $\varepsilon\leq 0$, an imputation $x$ is in the \emph{$\varepsilon$-core} of
$\Gamma$, if it satisfies $x(S) \geq v(S)+\varepsilon$ for all
$S \subset N$. Let
$$\varepsilon^*:=\mbox{sup} \{\varepsilon|\varepsilon\mbox{-core of} \
\Gamma \mbox{ is nonempty}\}.$$
The $\varepsilon^*$-core is called the least-core of $\Gamma$,
denoted by  ${\cal {LC}}(\Gamma)$, and the value $\varepsilon^*$ is
called the ${\cal {LC}} (\Gamma)$ value. Obviously,
the optimal solution of the following linear program $LP_1$
is exactly the value and the imputations in ${\cal LC}(\Gamma)$:
$$
LP_1:~\begin{array}{ll}
  \max & \quad \varepsilon  \\
 \mbox{s.t.} & \ \left\{ \begin{array}{ll}
 x(S)\geq v(S)+\varepsilon, \  & \  \forall S\subset N\\
 x_i\geq v(\{i\}), \  & \  \forall i\in N\\
 x(N)=v(N).  \\
\end{array} \right.
\end{array}
$$

Now we turn to the concept of the {\em nucleolus}.
Given any imputation $x$,  the excess of a coalition $S$ under $x$ is
defined as $e(x,S)=x(S)-v(S)$, which can be viewed as the satisfaction degree
of the coalition $S$ under the given $x$. The {\em excess vector} is the vector
$\theta(x)=(e(x,S_1),e(x,S_2),\cdots,e(x,S_{2^n-2}))$, where $S_1,\cdots,S_{2^n-2}$ is a list
of all nontrivial subsets of $N$ that satisfies
$e(x,S_1)\leq e(x,S_2)\leq \cdots \leq e(x,S_{2^n-2})$. The nucleolus of the
game $\Gamma$, denoted by $\eta(\Gamma)$,
is the imputation $x$ that lexicographically maximizes the excess vector $\theta(x)$.

Kopelowitz \cite{Kope:1967}, as well as Maschler, Peleg and Shapley \cite{Maschler:1979}
proposed that
$\eta(\Gamma)$ can be computed by recursively solving the following
sequential linear programs $SLP(\eta(\Gamma))$ $(k=1,2,\cdots)$:

$$LP_k:~\begin{array}{ll}
  \max & \quad \varepsilon  \\
 \mbox{s.t.} & \ \left\{ \begin{array}{ll}
 x(S)=v(S)+ \varepsilon_r, & \  \forall S\in {\cal J}_r \quad r=0,1,\cdots, k-1 \\
 x(S)\geq v(S)+\varepsilon,  & \  \forall  S\in 2^N \setminus\cup_{r=0}^{k-1}{\cal J}_r\\
 x_i\geq v(\{i\}), & \ \forall i\in N\\
 x(N)=v(N).  \\
\end{array} \right.
\end{array}$$
Initially, we set ${\cal J}_0=\{\emptyset, N\}$ and $\varepsilon_0=0$.
The number $\varepsilon_r$ is the optimal value of the $r$-th program $LP_r$,
and ${\cal J}_r=\{S\subseteq N: x(S)=v(S)+\varepsilon_r,\forall x\in X_r\}$,
where $X_r=\{x\in R^n:(x,\varepsilon_r)$ is an optimal solution of $LP_r\}$.
We call a coalition in ${\cal J}_r$ {\em fixed} since its allocation is fixed to a number.
Kopelowitz \cite{Kope:1967} showed that this procedure converges in at most $n$ steps.
Moreover, the nucleolus always exists and it is unique.
When ${\cal C}(\Gamma)\neq \emptyset$, $\eta(\Gamma)\in {\cal C}(\Gamma)$;
and $\eta(\Gamma)\in {\cal {LC}}(\Gamma)$, otherwise.

\begin{Proposition} {\emph{\cite{Elkind:2007,Osborne:1994}}}\label{lem:core}
A simple game $\Gamma=(N,v)$ has a nonempty core if and only if
there exists a veto player. Moreover,

(1) $x\in {\cal C}(\Gamma)$ if and only if $x_i=0$ for each $i\in N$
who is not a veto player;

(2) when ${\cal C}(\Gamma)\neq \emptyset$, the nucleolus of $\Gamma$
is given by $x_i=\frac{1}{k}$ if $i$ is a veto player and $x_i=0$ otherwise,
where $k$ is the number of veto players.
\end{Proposition}

\subsection{Threshold Matching Games}

We now introduce the definitions of threshold matching games.
For more detailed introduction, please refer to \cite{Aziz:2010, Elkind:2007,Kern:2003}.

Given a graph $G=(V,E)$, a matching $M$ is a set of
edges that no two edges in $M$ have a vertex in common.
The size of $M$ is denoted by $|M|$. A matching is maximum if
its size is maximum among all matchings in $G$. When there is a weight associated with each edge
$w: E\rightarrow R^+$, a matching $M$ is called a maximum
weight matching if its weight $w(M)=\sum_{e\in M}{w(e)}$
is maximum among all matchings.

For a weighted graph $G=(V,E;w)$ and a threshold $T\in R^+$, the corresponding
\emph{threshold matching game} (TMG) is a cooperative game defined as
$\Gamma=(V,w;T)$. We have the player set $V$ and $\forall\ S\subseteq V$,
\begin{center}
$v(S)\triangleq
\begin{cases}
1, & \text{if $w(M)\geq T$, where $M$ is the maximum weight matching of $G[S]$} \\
0, & \text{otherwise}
\end{cases}$
\end{center}
where $G[S]$ is the induced subgraph by $S$ on $G$.

Obviously, TMG is a simple game, and a player $i$ is
a veto player if and only if the weight of the maximum weight matching in
$G[V\setminus \{i\}]$ is less than $T$.
By Proposition~\ref{lem:core}, when there is a veto player, the core and
the nucleolus can be given directly. However, when the core is
empty, the least-core and the nucleolus is hard to compute\cite{Aziz:2010}.

\begin{Proposition}
{\emph{\cite{Aziz:2010}}}\label {general least-core and nucleolus}
Computing the least-core and nucleolus of $TMG$ is NP-hard if the core
of the TMG is empty.
\end{Proposition}

In the following we restrict ourselves to
\emph{threshold cardinality matching game} (TCMG) $\Gamma=(V;T)$
based on unweighted graph $G=(V, E)$. That is, $\forall S\subseteq N$,
$v(S)=1$ if the size of a maximum matching in $G[S]$ is no less than
$T$, and $v(S)=0$ otherwise.

Throughout the rest of this paper, we use the following notations:
\vspace{-1mm}
\begin{itemize}
  \item ${\cal M}$ : the set of matchings of $G$;
  \item ${\cal M}_T$ : the set of matchings $M_T$ of $G$ whose sizes are exact $T$, we call $M_T\in {\cal M}_T$ a minimal winning coalition;
  \item ${\cal M}^*$ : the set of matchings $M^*$ of $G$ with maximum size;
  \item $v^*$: the size of the maximum matching of $G$;
  \item $n$: $n=|V|$ is the number of players.
\end{itemize}

Let $G'=(V',E')$ be a subgraph of $G$. We usually use $i\in G'$ instead of $i\in V'$, and
$G'\setminus\{j\}$ instead of $G'[V'\setminus\{j\}]$.
For any imputation $x$, we define $x(G')\triangleq\sum_{i\in V'}x_i$
and $x(G'_{-j})\triangleq x(G'\setminus\{j\})=\sum_{i\in V',i\neq j}x_i$.
Specially, $\forall e=(i,j)\in E$, we let $x(e)=x(\{i,j\})=x_i+x_j$;
and $\forall M\in {\cal M}$, we let $x(M)=\sum_{e\in M}x(e)$.

Before entering into the details, we begin
with Gallai-Edmonds Decomposition of a graph, which plays an
important role in the nucleolus characterization.

\subsection{Gallai-Edmonds Decomposition}\label{matching theory}

Let $G=(V,E)$ be a graph.  A matching of $G$ is called a {\em perfect matching}
if it covers all vertices of $G$, and a {\em nearly perfect matching} if it
covers all vertices except one. $G$ is called \emph{factor-critical}
if removing any vertex of $G$, the rest graph has a perfect matching.

Given $A\subseteq V$, let $G\setminus A$ denote the subgraph induced by
vertices $V\setminus A$. $G\setminus A$ is composed of one or several maximal
connected components (hereinafter referred to as components). A component
of $G\setminus A$ is called \emph{even} (\emph{odd}) if it contains even
(odd) number of vertices. Denote by ${\cal B}={\cal{B}}(A)$ and
${\cal D}={\cal D}(A)$ the set of even components and odd components in
$G\setminus A$, respectively. We use $V(\cal{B})$ ($V(\cal{D})$)
to represent all vertices in even (odd) components.
A set $A\subseteq V$ is called a \emph{Tutte set}
if each maximum matching $M^*$ of $G$ can be decomposed as
$M^*=M_{\cal{B}}\cup M_{A,\cal{D}}\cup M_{\cal{D}}$,
where $M_{\cal{B}}$ induces a perfect matching in any even component
$B\in \cal{B}$, $M_{\cal{D}}$ induces a nearly perfect matching in any odd
component $D\in \cal{D}$, and $M_{A,\cal{D}}$ is a matching which matches
every vertex in $A$  to some vertex in an odd component in $\cal{D}$.
Thus, if $A$ is a Tutte set, the size $v^*$ of a maximum matching in
$G$ satisfies
$$v^*=\sum_{B\in\cal{B}}{\frac{|B|}{2}} + |A| +
\sum_{D\in\cal{D}}{\frac{(|D|-1)}{2}}.$$

\begin{Lemma}[Gallai-Edmonds Decomposition]
{\emph{\cite{Kern:2003,Plummer:1986,West:2011}}}\label{Tutte set}
Given $G=(V,E)$, one can construct a Tutte set $A\subseteq V$
in polynomial time such that
\vspace{-1mm}\begin{enumerate}
  \item all odd components $D\in\cal{D}$ are factor-critical;
  \item $\forall D\in\cal{D}$ there is a maximum matching $M^*$ of $G$
  which does not completely cover $D$ \emph{(}we say $M^*$ leaves $D$ \emph{uncovered}\emph{)}.
\end{enumerate}
\end{Lemma}

In the following we assume that $A\subseteq V$ is a fixed Tutte set satisfying
the condition 1 and 2 in Lemma \ref{Tutte set}. Thus, for such a Tutte set $A$,
we have the following facts:

1) For any $D\in\cal{D}$ and any maximum matching $M^*$,
$M^*$ matches at most one vertex in $A$ to $D$, hence, $|A|\leq |\cal{D}|$.
And when $\cal{D}\neq \emptyset$, $|A|<|\cal{D}|$.

2) For any Tutte set $A$ satisfying conditions in Lemma \ref{Tutte set},
the size of $A\cup V(\cal{B})$ is fixed which does not change with
the different choice of $A$.

    \section{Least-core of TCMG\label{sec:name:LC}}

In this section we firstly show that the least-core
problem for TCMG is much easier than that for TMG by providing a
polynomial time algorithm. Besides, we introduce a non-cooperative two-player zero-sum
game, called matching intercept game. There is a close relationship
between the mixed Nash equilibrium of this game and the least-core of TCMG.

Throughout this section, let $\Gamma=(V;T)$ be the TCMG defined on an
unweighted graph $G=(V,E)$ with threshold $T: 1\leq T \leq v^*$. Since
both testing the core nonemptiness and finding a core member can be done
efficiently, we focus on the case where ${\cal C}(\Gamma)=\emptyset$,
i.e., there is no veto players in $\Gamma$.

\subsection{Least-core}\label{generl-leastcore-reduce-result}\label{oracle}

For a TCMG $\Gamma=(V;T)$, denote ${\cal E}(\Gamma)=\{\{i\}:i\in V\}\cup \{\{M_T:M_T\in {\cal M}_T\}\}\cup \{V\}$.
We call a coalition $S\in {\cal E}(\Gamma)$ {\em essential coalition}.
We can show that for $\Gamma$, the least-core and nucleolus can be determined completely by the essential coalitions.

Suppose $S\subseteq V$ is a coalition of $\Gamma$.
If $S$ wins, then $S$ contains a minimal winning coalition $M_T\in {\cal M}_T$, and
$$x(S)-v(S)=x(M_T)-1+\sum_{i\in S\backslash V(M_T)}{x_i}\geq x(M_T)-1.$$
Since $x_i\geq 0$ for all $i\in V¡ä$, $S$ cannot be fixed before $M_T$ or any $i\in S\backslash V(M_T)$.
After $M_T$ and all $i\in S\backslash V(M_T)$ are fixed, $S$ is also fixed, i.e. coalitions like $S$ are redundant;
If $S$ loses, we can decompose $S$ in the following way:
$$x(S)-v(S)=\sum_{i\in V(S)}{x_i}\geq x_i,i\in S.$$
We can conclude that $S$ cannot be fixed before any $i\in V(S)$ through the same analysis. When all $i\in V(S)$ are fixed,
$S$ is fixed, i.e. $S$ is also redundant in this case.

From above analysis, we can conclude that the least-core ${\cal LC}(\Gamma)$ of TCMG can be
characterized as the optimal solution of the following linear program
$LP_1^T$:
$$
LP_1^T:~\begin{array}{ll}
  \max &  \varepsilon  \\
 \mbox{s.t.} & \ \left\{ \begin{array}{ll}
   x(M_T)\geq 1+ \varepsilon, \ & ~\forall M_T\in {\cal M}_T\\
   x(V)=1 \\
   x_i\geq 0, \ & ~i=1,2,\cdots,n.\\
\end{array} \right.
\end{array}
$$
We can obtain the same result in terms of the nucleolus, i.e., we have the following proposition.

\def\essentialcoalition{
Dropping the constraints associated with all $S\notin {\cal E}(\Gamma)$
will not change the result of $LP_1$ and $SLP(\eta(\Gamma))$.
}
\begin{Proposition}
    \label{essentialcoalitionpro}
    \essentialcoalition
\end{Proposition}

Therefore, when the threshold $T$ is a fixed number
independent of the input size, the size of the linear programs in $LP_1$ and $SLP(\eta(\Gamma))$
are all polynomial. It follows that the least-core and the nucleolus can be computed efficiently.
However, when the threshold $T$ is relevant to the input size,
the difficulty on computing is that each linear program in $LP_1$ and $SLP(\eta(\Gamma))$ remains
exponential size and we lack the understanding of the composition of ${\cal J}_r$.
In the following sections, we assume the threshold $T$ is relevant to the input size.

We firstly show that least-core can be solved efficiently
by ellipsoid method with a polynomial time separation oracle.

\def\leastcoreoracle{
Let $\Gamma=(V;T)$ be a TCMG with empty core. Then the problems of computing
the ${\cal LC}(\Gamma)$ value, finding a ${\cal LC}(\Gamma)$ member and checking
if an imputation is in ${\cal LC}(\Gamma)$ are all polynomial time solvable.
}
\begin{Theorem}\label{least-core-oracle}
     \leastcoreoracle
\end{Theorem}

\begin{proof}
A polynomial time separation oracle for $LP_1^T$ is as follows.
Let $(x,\varepsilon)$ be a candidate solution for $LP_1^T$.
Setting edge cost $\widetilde{c}(e)=x_i+x_j$ $(\forall e=(i,j)\in E)$
on the edge set of graph $G = (V,E)$.
Then we compute the minimum cost matching $M$ of size $T$.
If $x(M)\geq 1+\varepsilon$, then $(x,\varepsilon)$ is a feasible solution;
otherwise, the inequality $x(M)\geq 1+\varepsilon$ is a violated constraint.
Here, computing the minimum cost matching of fixed size can be done in polynomial time.\qed
\end{proof}

The above theorem shows that finding an imputation in the least-core
of TCMG is not as hard as that of the general case TMG.

In the following, we further provide a characterization of the least
core of TCMG under some conditions. Denote ${\cal M}_T=
\{M_1,M_2,\cdots,M_m\}$ be the set of all matchings whose sizes are exact $T$,
and let $a_1,a_2,\cdots,a_m$ be the
indicator vectors of the matchings in ${\cal M}_T$.
Consider the dual program of $LP_1^T$:

$$
DLP_1^T:~\begin{array}{ll}
  \min & \quad \delta-1  \\
  \mbox{s.t.} & \ \left\{ \begin{array}{ll}
  \displaystyle\sum_{M_j:i\in M_j}y_j\leq \delta,  \ &~i=1,2,\cdots,n\\
  \sum_{j=1}^{m}{y_j}=1 \\
  y_j\geq 0, \ &~j=1,2,\cdots,m.\\
\end{array} \right.
\end{array}
$$
Followed from the duality theorem, $DLP_1^T$ has the same optimal
value as $LP_1^T$. Hence, we have the following result which is quite
useful in the algorithm design for the nucleolus in next sections.

\def\mainresult{
Let $\Gamma=(V;T)$ be a TCMG with empty core.
If $(\frac{2T}{n},\cdots,\frac{2T}{n})_n$ is a convex combination of
$a_1,a_2\cdots,a_m$, then
\begin{enumerate}
  \item the value of ${\cal LC}(\Gamma)$ is $\varepsilon=\frac{2T}{n}-1$;
  \item $(\frac{1}{n},\cdots,\frac{1}{n})_n \in {\cal LC}(\Gamma)$;
  \item if there exists a convex combination and the coefficient corresponding to $a_i$ is strictly greater than zero, we have that the $i$-th constraint belongs to ${\cal J}_1$ in $SLP(\eta(\Gamma))$.
\end{enumerate}

}
\begin{Theorem}\label{general-T-dual-epsilon}
\mainresult
\end{Theorem}
\begin{proof}

The condition $(\frac{2T}{n},\cdots,\frac{2T}{n})_n$ is the convex
combination of $a_j$ $(j=1,2,\cdots,m)$ is equivalent to the fact that
there is a feasible solution to $DLP_1^T$ with the objective function value
being $\frac{2T}{n}-1$. On the other hand, it is easy to check $(x,\varepsilon)=
(\frac{1}{n},\cdots,\frac{1}{n},\frac{2T}{n}-1)$ is feasible to
$LP_1^T$. Followed from duality theorem of LP, $(x,\varepsilon)=
(\frac{1}{n},\cdots,\frac{1}{n},\frac{2T}{n}-1)$ is an optimal solution
of $LP_1^T$, yielding that $(\frac{1}{n},\cdots,\frac{1}{n})_n
\in {\cal LC}(\Gamma)$. By Complementary Slackness Theorem, we know if $a_i>0$, the $i$-th constraint is tight in each optimal solution. Thus, it  belongs to ${\cal J}_1$ in $SLP(\eta(\Gamma))$.
\qed
\end{proof}

\subsection{Matching Intercept Games}\label{MIG}
Now we define a non-cooperative zero-sum game on graph $G=(V,E)$
with two players, the \emph{interceptor} and the \emph{matcher}.
The pure strategy set of the interceptor is $V$ and
the pure strategy set of the matcher is ${\cal M}_T$.
If the vertex chosen by the interceptor intersects with the
matching chosen by the matcher, then the interceptor wins and gets payoff $1$;
otherwise, he loses and gets payoff $0$.
We call this non-cooperative game \emph{matching intercept game(MIG)}.

For the notion of {\em mixed Nash Equilibrium}, players select strategies
at random and act to maximum their own expected profit. In MIG,
let $p=(p_1,p_2,\cdots,p_n)$ be the interceptor's probability
distribution over his pure strategies, then based on Maxmin Theorem, the
optimal solution of the following linear program gives
the mixed Nash Equilibrium of MIG:
$$
\widetilde{LP}_1^T:~\begin{array}{ll}
  \max & \   \alpha  \\
 \mbox{s.t.} & \ \left\{ \begin{array}{ll}
   p(M_T)\geq  \alpha, \ & ~\forall M_T\in {\cal M}_T\\
   p(V)=1 \\
   p_i\geq 0, \ & ~i=1,2,\cdots,n.\\
\end{array} \right.
\end{array}
$$
It is obvious that linear program $\widetilde{LP}_1^T$ is equivalent to
linear program $LP_1^T$. Thus, for TCMG and MIG defined on the same graph $G$,
the least-core of TCMG is the same as the mixed Nash Equilibrium
of MIG.

%
%
%
%
%
%

In the following sections, we focus on the following three typical cases of TCMG to
give a clearer characterization on the nucleolus:
\vspace{-1mm}
\begin{itemize}
  \item Edge Coalitional Games (ECG): $T=1$ (section~\ref{sec:name:n1});
  \item TCMG on graphs with a perfect matching (section~\ref{sec:name:n2});
  \item TCMG on bipartite graphs (section~\ref{sec:name:n3}).
\end{itemize}
We aim to show that $\eta(\Gamma)$ is completely determined by edge coalitions (Coalition which contains exact two vertices and these two vertices form a edge in $G$ is called edge coalition.),
single player coalitions (which contains only one player)
as well as grand coalition rather than essential coalitions.

\section{Edge Coalitional Games (T=1)\label{sec:name:n1}}

In this section, we consider the edge coalitional game (ECG)
$\Gamma^1=(V;1)$ defined on an unweighted graph $G=(V,E)$, i.e., the TCMG with threshold $T=1$.
If there is a 0-degree vertex in $G$, then it has no contribution to any coalitions, i.e.,
its allocation must be 0. In the following, we assume that there is no 0-degree vertices in $G$.
It is easy to see that ${\cal C}(\Gamma^1)\neq \emptyset$ if and only if
there exists a vertex $i\in V$ such that there is no edges in $G\setminus \{i\}$,
that is, the graph $G$ is a star-like graph.

When ${\cal C}(\Gamma^1)=\emptyset$, the linear program for ${\cal LC}(\Gamma^1)$
is as follows:
$$
LP_1^1:~\begin{array}{ll}
  \max & \quad \varepsilon  \\
 \mbox{s.t.} & \ \left\{ \begin{array}{ll}
  x_i+x_j\geq 1+ \varepsilon, &~\forall e=(i,j)\in E\\
  x(V)=1 \\
  x_i\geq 0,&~ i=1,2,\cdots,n.\\
\end{array} \right.
\end{array}
$$

According to Gallai-Edmonds Decomposition, every graph can be decomposed into
$A$, $\cal B$, $\cal D$. Let ${\cal D}_0$ be  the set of singletons in $\cal D$
(${\cal D}_0$ may be empty).
Since ${\cal D}_0$ is the set of singletons, we ambiguously use ${\cal D}_0$ instead of $V({\cal D}_0)$.
Let $G_0$ be a bipartite graph  with vertex set $A \cup {\cal D}_0$ and edge set
consisting of edges with two endpoints in $A$ and ${\cal D}_0$ separately.
Find a maximum matching $M_0$ in $G_0$. Denote the matched vertices
in $A$ and ${\cal D}_0$ by $A_1$ and ${\cal D}_{01}$ with respect to $M_0$.
Let $A_2 = A\setminus A_1$ and ${\cal D}_{02} = {\cal D}\setminus {\cal D}_{01}$.

We firstly consider the simple case that ${\cal D}_{02}=\emptyset$, that is,
all singletons in $\cal D$ can be matched to $A$. By making use of
Theorem \ref{general-T-dual-epsilon}, we show that in this case,
the least-core value and an imputation in the least-core can be obtained directly.
Then we generalize the result into the case  ${\cal D}_{02}\neq \emptyset$.

\def\edgeproposition{
Given an ECG $\Gamma^1=(V;1)$, if ${\cal D}_{02}=\emptyset$, then the value
of ${\cal LC}(\Gamma^1)$ is $\varepsilon=\frac{2}{n}-1$ and $(\frac{1}{n},\cdots,\frac{1}{n})_n\in{\cal LC}(\Gamma)$.
}

\begin{Theorem}\label{general-leastcore-T=1}
\edgeproposition
\end{Theorem}

\begin{proof}
By Theorem \ref{general-T-dual-epsilon}, it is enough to show that $(\frac{2}{n},\cdots,\frac{2}{n})$
is the convex combination of the coefficients of $x_i+x_j\geq 1+\varepsilon,e=(i,j)\in E$.

Since ${\cal D}_{02}=\emptyset$, let $M$ be a maximum matching in $G$ such that $M$ matches all ${\cal D}_{0}$ into $A$.
Then ${\cal D}_{0}\subseteq{\cal D}_1\triangleq\{D\in {\cal D}:D \mbox{ is covered by} ~M \}$
and ${\cal D}_2 \triangleq {\cal D}\setminus {\cal D}_1$ is the set of uncovered factor-critical graphs.
So for any $D_i\in {\cal D}_2$, we have $n_i=|D_i|\geq 3$.
Let $\widetilde{M}$ is the induced matching from $M$ in $V({\cal B}) \cup A \cup V({\cal D}_1)$.
Then $\widetilde{M}$ is a perfect matching in $G[V({\cal B}) \cup A \cup V({\cal D}_1)]$.
For any edge $e\in \widetilde{M}$, we set $\frac{2}{n}$ to the element in the convex combination corresponding to this constraint.
There are in total $\frac{|V({\cal B}) \cup A \cup V({\cal D}_1)|}{2}$ constraints,
and each vertex in $V({\cal B}) \cup A \cup V({\cal D}_1)$ appears exactly once in these constraints.
For any $D_i\in{\cal D}_2$, since it is factor-critical, then for every vertex $k\in V(D_i)$,
there exists a perfect matching $M_k$ in $G[V(D_i)\setminus \{k\}]$.
For any edge $e\in M_k$, we set $\frac{2}{(n_i-1)n}$ to the element in the convex combination corresponding to this constraint.
When we traverse all vertices $k\in V(D_i)$, there are in total $\frac{n_i (n_i-1)}{2}$ constraints,
and for each vertex $k$, it appears in exactly $n_i-1$ constraints.
Set the elements in
the convex combination corresponding to all other constraints to be 0.
It is easy to check the convex combination of these constraints is $(\frac{2}{n},\cdots,\frac{2}{n})$.
This finishes our proof.
\qed
\end{proof}

If ${\cal D}_{02}\neq \emptyset$, we cannot find such a convex combination since
there is no edges in ${\cal D}_{02}$. But if we delete ${\cal D}_{02}$ from $G$,
we can find a convex combination in $G'=G[V\setminus {\cal D}_{02}]$ by using the similar argument in the proof of Theorem \ref{general-leastcore-T=1}.
Denote $\Gamma'$ to be the corresponding ECG defined on $G'$. Then the value of ${\cal LC}(\Gamma')$ is $\frac{2}{n'}-1$ where  $n'=n-|{\cal D}_{02}|$.
We can know $\frac{2}{n'}-1$ is an upper bound of the value of ${\cal LC}(\Gamma^1)$.
If we can find a feasible solution $(\widetilde{x},\varepsilon)$ of the original game
with $\varepsilon = \frac{2}{n'}-1$, then the value of ${\cal LC}(\Gamma^1)$ is also $\frac{2}{n'}-1$.
Consider the following imputation:
$$\widetilde{x}_i=
\begin{cases}
\frac{1}{n'}, & i\in V({\cal B}) \\
\frac{2}{n'}, & i\in V(A)\ \mbox{and} \ {\cal D}_{02}\rightarrow i\\
\frac{1}{n'}, & i\in V(A)\ \mbox{and} \  {\cal D}_{02}\nrightarrow i\\
0, & i\in V({\cal D}_{01})\ \mbox{and} \  {\cal D}_{02}\rightarrow i\\
\frac{1}{n'}, & i\in V({\cal D}_{01})\ \mbox{and} \  {\cal D}_{02}\nrightarrow i\\
0, & i\in V({\cal D}_{02}).
\end{cases}$$

Here, ${\cal D}_{02}\rightarrow i$ represents $i$ is reachable from ${\cal D}_{02}$
by $M_0$-alternating path in $G_0$;
${\cal D}_{02}\nrightarrow i$ represents $i$ is unreachable from ${\cal D}_{02}$
by $M_0$-alternating path in $G_0$.

We can easily check that the imputation $\widetilde{x}$ with $\varepsilon=\frac{2}{n'}-1$
is feasible. Otherwise, if there is an edge $e\in E$ such that
$\widetilde{x}(e) < \frac{2}{n'}$, $e$ must be between a vertex $i\in{\cal D}_{01}$ (and ${\cal D}_{02}\rightarrow i$)
and a vertex in $A_2$ or between a vertex in ${\cal D}_{02}$ and a vertex in $A_2$.
But there exists a $M_0$-augmenting path in $G_0$ in these two cases, contradicting to
$M$ is a maximum matching in $G_0$. Hence, the value of ${\cal LC}(\Gamma^1)$ is $\frac{2}{n'}-1$.

We then focus on the computing of nucleolus.
Since we have seen $\varepsilon_1=\frac{2}{n'}-1$, we can prove $LP_k^1$ in $SLP(\eta(\Gamma^1))$ can be rewritten as:

$$
LP_k^1:~\begin{array}{ll}
  \max & \quad \varepsilon  \\
 \mbox{s.t.} & \ \left\{ \begin{array}{ll}
  x(e)= \frac{2}{n'}-\varepsilon_1+\varepsilon_r, &~e\in E_r,r=1,\cdots,k-1\\
  x_i = -\varepsilon_1+\varepsilon_r, &~i\in V_r,r=1,\cdots,k-1\\
  x(e)\geq \frac{2}{n'}-\varepsilon_1+\varepsilon, &~e\in E\setminus \bigcup_{r=1}^{k-1}E_r\\
  x_i\geq -\varepsilon_1+\varepsilon, &~i\in V\setminus \bigcup_{r=1}^{k-1}V_r\\
  x(V)=1,~x_i\geq 0, &~i\in V.\\
\end{array} \right.
\end{array}
$$
Initially set $E_0=V_0=\emptyset$ and $\varepsilon_0=0$.
The number $\varepsilon_r$ is the optimal value of the $r$-th program $LP_r^1$,
and $E_r=\{e\in E: x(e)=1+\varepsilon_r,\forall x\in X_r\}$,
$V_r=\{i\in N: x_i=1-\frac{2}{n'}+\varepsilon_r,\forall x\in X_r\}$,
where $X_r=\{x\in R^n:(x,\varepsilon_r)$ is an optimal solution of $LP_r^1\}$.

In the next sections, we will show that for any threshold, $LP_k^T$ have the same appearance as $LP_k^1$ under some restrictions.


\def\nucleolus_T_1{
Given an ECG $\Gamma^1=(V;1)$,
the nucleolus $\eta(\Gamma^1)$ can be obtained in polynomial time.
}

\section{TCMG on Graph with Perfect Matching\label{sec:name:n2}}

Now we consider the general case $\Gamma^T=(V;T)$ with arbitrary threshold $1\leq T\leq v^*$.
We denote the corresponding sequential linear programming as $LP_k^T$.

In the following theorem, we firstly show that for the graphs with a perfect matching, the least-core of $\Gamma^T$ is independent of $T$.
Then we use this characterization to prove that the nucleolus of $\Gamma^T$ can be obtained in polynomial time and $\eta(\Gamma^T)$ is also independent of $T$,
i.e., $LP_k^T$ can be rewritten as:
\begin{equation}\label{perfect-LP}
LP^T_k:~\begin{array}{ll}
  \max & \quad \varepsilon^T  \\
 \mbox{s.t.} & \ \left\{ \begin{array}{ll}
  x(e)=\frac{2}{n}-\varepsilon^T_1 +\varepsilon_r^T, &~e\in E_r,r=1,\cdots,k-1\\
  x_i= -\varepsilon_1^T+\varepsilon^T_r, &~i\in V_r,r=1,\cdots,k-1\\
  x(e)\geq \frac{2}{n}-\varepsilon^T_1+\varepsilon^T, &~e\in E\setminus \bigcup_{r=1}^{k-1}E_r\\
  x_i\geq -\varepsilon^T_1+\varepsilon^T, &~i\in V\setminus \bigcup_{r=1}^{k-1}V_r\\
  x(V)=1,~x_i\geq 0, &~i\in V.\\
\end{array} \right.
\end{array}
\end{equation}
Initially, we set $E_0=V_0=\emptyset$ and $\varepsilon_0^T=0$.
The number $\varepsilon_r^T$ is the optimal value of the $r$-th program $LP_r^T$,
and $E_r=\{e\in E: x(e)=1+\varepsilon_r^T,\forall x\in X_r\}$,
$V_r=\{i\in N: x_i=1-\frac{2}{n}+\varepsilon_r^T,\forall x\in X_r\}$,
where $X_r=\{x\in R^n:(x,\varepsilon_r^T)$ is an optimal solution of $LP_r^T\}$.

\def\lemmratioandpoint{
Suppose $G=(V,E)$  is a simple graph which has a perfect matching and $\Gamma^T=(V;T)$ is a TCMG defined on $G$. Let $\Gamma^1 = (V;1)$ be the corresponding ECG defined also on $G$.
Then
\begin{enumerate}
  \item the value of ${\cal LC}(\Gamma^T)$ is $\varepsilon^T_1=\frac{2T}{n}-1$;
  \item ${\cal LC}(\Gamma^T) = {\cal LC}(\Gamma^1)$;
  \item $\eta(\Gamma^T)=\eta(\Gamma^1)$.
\end{enumerate}
}

\begin{Theorem}
     \label{perfect-proposition}
     \lemmratioandpoint
\end{Theorem}
\begin{proof}
\

1. By Theorem \ref{general-T-dual-epsilon}, it is enough to show that $(\frac{2T}{n},\cdots,\frac{2T}{n})$
is the convex combination of the coefficients of $x(M_T)\geq 1+\varepsilon,M_T\in {\cal M_T}$.
Since $G$ has a perfect matching $M^*$, then $v^*=\frac{n}{2}$.
Without loss of generality, we reset the labels of players as
$(1,2),\cdots,(n-1,n)$, here $(i,j)$ means it is a matching edge in $M^*$.
Then there are constraints like the following in $LP_1^T$
\begin{equation}\label{perfect-T-plus}
\begin{aligned}
  & (x_1+x_2)+\cdots+(x_{2T-1}+x_{2T})\geq 1+\varepsilon^T  \\
  & (x_3+x_4)+\cdots+(x_{2T+1}+x_{2T+2})\geq 1+\varepsilon^T\\
  & \cdots\\
  & (x_{n-1}+x_n)+\cdots+(x_{2T-3}+x_{2T-2})\geq 1+\varepsilon^T.\\
\end{aligned}
\end{equation}
We put all the constraints above into the combination with an element $\frac{2}{n}$.
There are in total $\frac{n}{2}$ constraints, and each vertex in $V$ appears exactly $T$
times in these constraints.
It is easy to check this convex combination of these constraints is $(\frac{2T}{n},\cdots,\frac{2T}{n})$.

2. Suppose $x=(x_1,\cdots,x_n)$
with $\varepsilon^T=\frac{2T}{n}-1$ is an optimal solution in $LP^T_1$
(the first linear program in $SLP(\eta(\Gamma^T))$),
 $e=(i,j)$ is a maximum matching edge in $M^*$.
Since constraints in (\ref{perfect-T-plus}) are fixed, then
\begin{equation}\label{perfec-M}
x_i+x_j= \frac{2}{n},~~~\forall e=(i,j)\in M^*.
\end{equation}

If $e=(i,j)$ is not an edge of $M^*$, $e$ must be in
a matching $M'$ with size $v^*-1$ and all the other $v^*-2$ edges except $e$
are belonging to $M^*$, i.e., those edges are fixed to $\frac{2}{n}$.
Since for all $M_T\in \mathcal{M_T}$, $x(M_T)\geq 1+\varepsilon^T=\frac{2T}{n}$, we have
\begin{equation}\label{perfec-E-M}
	x_i+x_j\geq \frac{2}{n},\qquad \forall~e=(i, j)\in E\setminus M^*.
\end{equation}
From (\ref{perfec-M}) and (\ref{perfec-E-M}), we can see that
$(x,\frac{2}{n}-1)$ is also an optimal solution of $LP^1_1$.

On the other hand,
let $x'=(x'_1,\cdots,x'_n)$ with $\varepsilon^1_1=\frac{2}{n}-1$ be an
optimal solution in $LP^1_1$.
We can quickly check $x'$ with $\varepsilon_T=\frac{2T}{n}-1$ is an optimal solution in $LP^T_1$.

Therefore, ${\cal LC}(\Gamma^T)={\cal LC}(\Gamma)$.

3. Suppose the set of fixed constraints in $LP^T_1$ is ${\cal M}'_T\subseteq {\cal M}_T$
(and the set of fixed coalitions is also ${\cal M}'_T$) and
$E_1$ is the fixed constraints of $LP^1_1$. Firstly, we prove that
$E_{\cal T}=E_1$, $E_{\cal T}=\{e|e\in M_T\in {\cal M}'_T\}$.

Since the system of linear equations
$x(M_T)=\frac{2T}{n},M_T\in {\cal M}'_T$ is equivalent to the system of linear equations
$x(e)=\frac{2}{n},e\in E_T$.
Otherwise, if there exists some $e\in M_T\in {\cal M}'_T$ with $x(e)>\frac{2}{n}$,
$x(M_T)$ cannot equal to $\frac{2T}{n}$,
due to $x(e)\geq \frac{2}{n},\forall e\in E$.
Because ${\cal LC}(\Gamma^T)={\cal LC}(\Gamma)$, $E_{\cal T}=E_1$.

Then we simplify the sequence of linear programs.

\texttt{Case} 1:  Consider the winning constraints like
$$x(S)=\sum_{e\in E'}{x(e)}+\sum_{e\in E''}{x(e)}+\sum_{i\in V'}{x_i} \geq 1+\varepsilon^T.$$
Here, $E'\subseteq E,E'\cap E_1=\emptyset$, $E''\subseteq E_1$ and $V'\subseteq V$.
The size of the maximum matching in $S=V(E')\cup V(E'')\cup V'$ is not less than $T$ and suppose $|E'|\geq 2$.
It will be fixed after any
\begin{equation}\label{more-2-unfixed-1}
x(e)+\sum_{e\in S'\subseteq E_1,~|S'|=T-1}{x(e)}\geq 1+\varepsilon^T,~e\in E'
\end{equation}
and
\begin{equation}\label{more-2-unfixed-2}
x_i+\sum_{e\in S''\subseteq E_1,~|S''|=T}{x(e)}\geq 1+\varepsilon^T,~i\in V',
\end{equation}
since the excess of the coalition $S$ is not less than
any coalition $\{e\}\cup S'$ or $\{i\}\cup S''$.
Moreover, it will be fixed automatically after
all constraints like (\ref{more-2-unfixed-1}) and (\ref{more-2-unfixed-2}) get fixed.

Due to result 1 and 2 above,
we can rewrite (\ref{more-2-unfixed-1}) and (\ref{more-2-unfixed-2})
as
\begin{equation}\label{x_e-in-perfect-nucleolus-formal}
x(e)\geq \frac{2}{n}-\varepsilon^T_1+\varepsilon^T,~e\in E\setminus E_1,
\end{equation}
\begin{equation}\label{x_i-in-perfect-nucleolus-formal}
x_i\geq -\varepsilon^T_1+\varepsilon^T,~i\in V.
\end{equation}

\texttt{Case} 2: Consider the losing constraints like
$$x(S)=\sum_{e\notin E_1}{x(e)}+\sum_{e\in E_1}{x(e)}+\sum{x_i} \geq \varepsilon^T.$$
The maximum matching in $S$ is less than $T$ and it will be fixed after
(\ref{x_e-in-perfect-nucleolus-formal}) and
(\ref{x_i-in-perfect-nucleolus-formal}), since the excess of the coalition $S$ is not less than
any subset in $S$.
Moreover, it will be fixed automatically after
all constraints like (\ref{x_e-in-perfect-nucleolus-formal})
and (\ref{x_i-in-perfect-nucleolus-formal}) get fixed.

Hence, the sequential linear programs $LP^T_k$ of $\eta(\Gamma^T)$ can be rewritten as
linear program (\ref{perfect-LP}).
It is obvious
that the optimal solutions of  $LP^T_k$ and $LP^1_k$ are the same
except $\varepsilon^T_k=\varepsilon^1_k+\frac{2(T-1)}{n}$ before
$\varepsilon_k$ gets positive, i.e., $\eta(\Gamma^T)=\eta(\Gamma)$.

This finishes our proof.
\qed

\end{proof}

    \section{TCMG on Bipartite Graphs\label{sec:name:n3}}

For bipartite graphs, we can obtain the similar result as Theorem \ref{perfect-proposition}.
Let $G=(L,R;E)$ be a bipartite graph with vertex set $L\cup R$ and edge set $E$.
Find a maximum matching $M^*$ in $G$. Denote the matched vertices
in $L$ and $R$ as $L_1$ and $R_1$ with respect to $M^*$ respectively.
Let $L_2 = L\setminus L_1$ and $R_2 = R\setminus R_1$.

If both $L_2$ and $R_2$ are empty, it is reduced to the situation in section \ref{sec:name:n2}.
So we assume at least one of $L_2$ and $R_2$ is not empty.

If we delete $L_2$ and $R_2$ from $G$,
we can find the least-core value and an imputation in least-core by Theorem \ref{perfect-proposition}.
Denote $\Gamma'$ to be the corresponding TCMG defined on $G'$ where $G'$ is the induced subgraph by $(L\cup R)\setminus(L_2\cup R_2)$ in $G$.
Then the value of ${\cal LC}(\Gamma')$ is $\frac{2T}{n'}-1$ where  $n'=n-|L_2|-|R_2|$.
It is obvious that $\frac{2T}{n'}-1$ is an upper bound of the value of ${\cal LC}(\Gamma^T)$.
We then show that this is actually the value of the least-core in the bipartite graphs.

\def\bipartitegraph{
Suppose $G=(L,R;E)$  is a bipartite graph and $\Gamma^T=(V;T)$ is a TCMG defined on $G$. Let $\Gamma^1 = (V;1)$ be the corresponding ECG defined also on $G$.
Then
\begin{enumerate}
  \item the value of ${\cal LC}(\Gamma^T)$ is $\varepsilon^T_1=\frac{2T}{n'}-1$;
  \item ${\cal LC}(\Gamma^T) = {\cal LC}(\Gamma^1)$;
  \item $\eta(\Gamma^T)=\eta(\Gamma^1)$.
\end{enumerate}
Here $n'=n-|L_2|-|R_2|$.
}

\begin{Theorem}
     \label{the_bipartitegraph}
     \bipartitegraph
\end{Theorem}

\begin{proof}
We only prove the fist result. The second and the third ones are the same as Theorem \ref{perfect-proposition}.

If we can find a feasible solution $(\widetilde{x},\varepsilon)$ of the original game
with $\varepsilon = \frac{2T}{n'}-1$, then the value of ${\cal LC}(\Gamma^1)$ is also $\frac{2T}{n'}-1$.
Here, we use $L_2\rightarrow i,i\in L_1\cup R_1$ to represent that $i$ is reachable from $L_2$
by $M^*$-alternating path in $G$. We denote these vertices which are in $L_1$ by $L_{11}$ and the vertices which are in $R_1$ by $R_{11}$;
Similarly, $R_2\rightarrow i,i\in L_1\cup R_1$ represents $i$ is reachable from $R_2$
by $M^*$-alternating path in $G$, and we denote these vertices which are in $L_1$ by $L_{12}$ and the vertices which are in $R_1$ by $R_{12}$.
Denote $L_{13}=L_1\backslash (L_{11}\cup L_{12})$, $R_{13}=R_1\backslash(R_{11}\cup R_{12})$.
Note that $L_{11}\cap L_{12}=\emptyset$ and $R_{11}\cap R_{12}=\emptyset$.
Otherwise, without loss of generality, we assume $k\in L_{11}\cap L_{12}$.
Then, there exists two $M^*$-alternating paths in $G$, $P_1:l_2 \rightarrow k, l_2\in L_2$ and $P_2:r_2\rightarrow k, r_2\in R_2$.
Find $k'$ to be the first intersection point of $P_1$ and $P_2$. Here first intersection point means there is no other intersection points
locating in $l_2 \xrightarrow{P_1} k'$ or $r_2 \xrightarrow{P_2} k'$.
We know that one of the two edges incidence of $k'$ in $P_1$ and $P_2$ is matched edge and the other is unmatched.
Then we find a $M^*$-augmenting path in $G$:
$l_2 \xrightarrow{P_1} k' \xrightarrow{P_2} r_2$,
contradicting to $M^*$ is a maximum matching .
Hence such $k$ does not exist.

Consider the following imputation:
$$\widetilde{x}_i=
\begin{cases}
0, & i\in L_2\cup R_2\\
\frac{2}{n'}, & i\in R_{11} \cup L_{12} \\
0, & i\in L_{11}\cup R_{12} \\
\frac{1}{n'}, & i\in L_{13}\cup R_{13}\\
\end{cases}$$

Therefore our imputation $\widetilde{x}$ is well defined.
Then we can easily check that the imputation $\widetilde{x}$ with $\varepsilon=\frac{2T}{n'}-1$ is feasible like section \ref{sec:name:n1}.
In fact, we can find the union of $L_{11}$ and $R_{12}$ (or $L_{12}$ and $R_{11}$) is a Tutte set in Gallai-Edmonds Decomposition.
\qed
\end{proof}

    \section{Conclusion \label{sec:name:Conclusion}}

In this paper, we first design a polynomial time algorithm to compute 
the least-core for threshold cardinality matching games.
Based on a general characterization of the least-core for TCMG, 
we show that computing the nucleolus can be done efficiently for TCMGs  
defined on graphs with a perfect matching and bipartite graphs.

We conjecture that the ideas behind these results can be generalized to 
compute the nucleolus of TCMGs defined on general graphs. 
Another interesting direction is to understand how far can these 
techniques be extended to the computation of the least-core and the 
nucleolus of other threshold versions of cooperative games~\cite{Aziz:2010}.


    \bibliography{ref}

\end{document}